\documentclass[runningheads]{llncs}
\usepackage{graphicx}
\usepackage{latexsym,amssymb,amsmath}
\usepackage{multirow}

\newcommand{\lii}{[\![}

\newcommand{\rii}{]\!]}

\usepackage{lipsum}

\newcommand\blfootnote[1]{%
  \begingroup
  \renewcommand\thefootnote{}\footnotetext{#1}%
  \addtocounter{footnote}{-1}%
  \endgroup
}

\begin{document}
\title{On decoding hyperbolic codes}
%
%

\author{Eduardo Camps-Moreno$\,^{1}$  \and Ignacio Garc\'ia-Marco$\,^{2}$ \and Hiram H. L\'opez$\,^{1}$ \and \\ Irene M\'arquez-Corbella$\,^{2}$ \and Edgar Mart\'inez-Moro$\,^{3}$ \and Eliseo Sarmiento$\,^{4}$}
\authorrunning{E. Camps-Moreno et al.}
%

\institute{
 Department of Mathematics, Virginia Tech, Blacksburg, VA, USA \\ \email{\{eduardoc, hhlopez\}@vt.edu} \and 
Facultad de Ciencias and Instituto de Matem\'aticas y Aplicaciones (IMAULL).  
Universidad de La Laguna, Spain\\ \email{\{iggarcia,imarquec\}@ull.es} \and
 Institute of Mathematics, University of Valladolid, Castilla Spain,\\ \email{edgar.martinez@uva.es}\and 
  Instituto Polit\'ecnico Nacional,
 Mexico\\
\email{esarmiento@ipn.mx}
}
\maketitle              
\begin{abstract}
This work studies several decoding algorithms for hyperbolic codes.  We use some previous ideas to describe how to decode a hyperbolic code using the largest Reed-Muller code contained in it or using the smallest Reed-Muller code that contains it. A combination of these two algorithms is proposed when hyperbolic codes are defined by polynomials in two variables. Then, we compare hyperbolic codes and Cube codes (tensor product of Reed-Solomon codes) and propose decoding algorithms of hyperbolic codes based on their closest Cube codes. Finally, we adapt to hyperbolic codes the Geil and Matsumoto's generalization of Sudan's list decoding algorithm. 
\blfootnote{E. Camps-Moreno, H.~H.~L\'opez, E.~Mart\'inez-Moro and I. M\'arquez-Corbella were partially supported by Grant TED2021-130358B-I00 funded by MCIU/AEI/10.13039/501100011033 and by the "European Union NextGenerationEU/PRTR''. \\  H. H. L\'opez was partially supported by the NSF grants DMS-2201094 and DMS-2401558.\\
I. Garc\'ia-Marco and I. M\'arquez-Corbella were partially supported by the Spanish MICINN PID2019-105896GB-I00
}
\keywords{Reed-Muller codes\and evaluation codes \and hyperbolic codes \and decoding algorithms.}
\end{abstract}

\section{Introduction}
Let $\mathbb F_q$ be a finite field with $q$ elements. Given two vectors $\mathbf x=(x_1,\ldots, x_n)$ and $\mathbf y=(y_1,\ldots, y_n)\in \mathbb F_q^n$, the \emph{Hamming distance} between $\mathbf x$ and $\mathbf y$ is defined as $d_H(\mathbf x, \mathbf y) = |\{ i \mid x_i \neq y_i\}|$, where $|\cdot |$ denotes the cardinality of the set. The \emph{Hamming weight} of $\mathbf x$ is given by $\mathrm{w}_H(\mathbf x)=d_H(\mathbf x, \mathbf 0),$ where $\mathbf 0$ denotes de zero vector in $\mathbb F_q^n.$ The \emph{support} of $\mathbf x$ is the set $\mathrm{supp}(\mathbf x)=\{ i \mid x_i \neq 0\}.$ An $[n(C),k(C),\delta(C)]_q$ {\it linear code} $C$ over $\mathbb{F}_q$ is an $\mathbb{F}_q$-vector space of $\mathbb{F}_q^{n(C)}$ with dimension $k(C)$, and minimum distance $\delta(C)=\min \{ d_H(\mathbf{c},\mathbf{c}'): \mathbf{c},\mathbf{c}' \in C, \mathbf{c} \neq \mathbf{c}' \},$ thus  its error correction capability is $t_{\mathcal C}=\left \lfloor{\frac{\delta(C)-1}{2}}\right \rfloor$. When there is no ambiguity, we write $n, k$, $\delta,$ and $t$ instead of $n(C),k(C)$, $\delta(C),$ and $t_{\mathcal C}$, respectively.

Let $\mathbb N$ be the set of non-negative integers. 
For $A\subseteq \mathbb N^m$, 
$\mathbb F_q[A]$ is the $\mathbb F_q$-vector subspace of polynomials in $\mathbb F_q[\mathbf X] = \mathbb F_q[X_1, \ldots, X_m]$ with basis given by the set of monomials
$\left\{ \mathbf{X}^{\mathbf i} = X_1^{i_1} \cdots X_m^{i_m} \mid \mathbf i = (i_1, \ldots, i_m) \in A\right\}$.
Let $\mathcal P = \mathbb F_q^m$, where $ \mathcal P = \{P_1, \ldots, P_n\}$ and  $n=|\mathcal P| = q^m.$ Define the following evaluation map
\[\begin{array}{cccc}
\mathrm{ev}_{\mathcal P}: & \mathbb F_q[X_1, \ldots, X_m] & \longrightarrow & \mathbb F_q^n\\
& f & \longmapsto & (f(P_1), \ldots, f(P_n)).
\end{array}\]
The {\it monomial code associated} to $A,$ denoted by $\mathcal C_A,$ is defined as
$$\mathcal C_A = \mathrm{ev}_{\mathcal P} (\mathbb F_q[A]) = \left\{\mathrm{ev}_{\mathcal P}(f) \mid f\in \mathbb F_q[A] \right\}.$$
For $a, b \in \mathbb R$ and $a \leq b$, we denote by $\lii a, b \rii$ the integer interval $[a,b] \cap \mathbb Z$.

\begin{definition}[Reed-Muller and Hyperbolic  codes]
\begin{itemize}\item Let $s \geq 0, m \geq 1$ be integers. The monomial code $\mathcal C_{R}$  where the set $R$ is given by  $R= \left\{\mathbf i = (i_1, \ldots, i_m) \in \lii 0,q-1 \rii^m \mid \sum_{j=1}^m i_j \leq s \right\}$, is called the \emph{Reed-Muller code} over $\mathbb{F}_q$ of degree $s$ with $m$ variables. This code is denoted by $\mathrm{RM}_q(s,m).$
\item
Let $d, m \geq 1$ be integers. The monomial code $\mathcal C_H$  where the set $H$ is given by 
$H=\left\{ \mathbf i = (i_1, \ldots, i_m) \in \lii 0,q-1\rii^m \mid \prod_{j=1}^m (q-i_j) \geq d \right\}$,
is called the \emph{hyperbolic code} over $\mathbb{F}_q$ of order $d$ with $m$ variables. This code is denoted by $\mathrm{Hyp}_q(d,m).$\end{itemize}
\end{definition} 
In our previous work \cite{our} we proved there are two optimal Reed-Muller codes such that   $\mathrm{RM}_q(s,m) \subseteq \mathrm{Hyp}_q(d,m) \subseteq \mathrm{RM}_q(s^\prime,m)$. In other words,   the largest Reed-Muller code $\mathrm{RM}_q(s,m)$ contained in $\mathrm{Hyp}_q(d,m),$ and the smallest Reed-Muller code $\mathrm{RM}_q(s^\prime,m)$ that contains $\mathrm{Hyp}_q(d,m).$  We will use that result to propose several decoding procedures.

The paper is organized as follows. In Section~\ref{21.04.04}, we describe two different algorithms to decode a hyperbolic code $\mathrm{Hyp}_q(d,m)$. These algorithms are based on the optimal Reed-Muller code that approximates to our hyperbolic code, that is, the largest Reed-Muller code $\mathrm{RM}_q(s,m)$ contained in $\mathrm{Hyp}_q(d,m)$, or the smallest Reed-Muller code $\mathrm{RM}_q(s^\prime,m)$ that contains $\mathrm{Hyp}_q(d,m)$. We will study the advantages and disadvantages in terms of efficiency and correction capability of these proposed algorithms. The choice of the algorithm to be used depends on which Reed-Muller code is closest to the hyperbolic code $\mathrm{Hyp}_q (d, m)$ as well as the efficiency or effectiveness that we need. At the end of that section, a third algorithm, which is a combination of the previous two, is adapted for $m=2,$ the case of two variables. In Section~\ref{21.04.05}  a decoding algorithm for a hyperbolic code $\mathrm{Hyp}_q(d,m)$ based on the tensor product of Reed-Solomon codes is presented. In Section~\ref{21.04.06}
 we recover Geil and Matsumoto's generalization of Sudan's List Decoding for order domain codes (\cite{GM07}) but focus explicitly on hyperbolic codes. The novel idea that we present here is that we explicitly describe each of the sets involved in the algorithm or at least we give a subset of such sets. Finally in Section~\ref{sec:con}, we compare the performance of the five decoding algorithms proposed in this paper.

 \section{Decoding  based on known Reed-Muller decoders}\label{21.04.04}
\subsection{Decoding algorithm based on the smallest Reed-Muller code}
\label{sub:dec1}
The main idea that we present in this section is well-known and works for any pair of nested linear codes.
Let $\mathcal C_1 \subseteq \mathcal C_2\subseteq \mathbb F_q^n$ be two linear codes with parameters $[n,k_1, d_1]_q$ and $[n,k_2,d_2]_q$, respectively. Observe that a decoding algorithm for $\mathcal C_2$ that corrects up to $t_{2}$ errors is also a decoding algorithm for $\mathcal C_1$ that requires the same number of operations as in $C_2,$ but corrects up to $t_2$ errors. Note that $d_2 \leq d_1$ and the difference between these two values might be huge. That is, a decoding algorithm for $\mathcal C_2$ is also a decoding algorithm for $\mathcal C_1,$ but there is a loss in the number of errors that one might expect to correct.

Given a hyperbolic code $\mathrm{Hyp}_q(d,m)$, by  \cite[Theorem~3.6]{our} we know the smallest integer $s$ such that $\mathrm{Hyp}_q(d,m) \subseteq \mathrm{RM}_q(s,m).$ Thus, for each decoding algorithm for $\mathrm{RM}_q(s,m)$ we have one for $\mathrm{Hyp}_q(d,m)$. For example, if we use the already mentioned Pellikaan-Wu's list-decoding algorithm, one can correct up to 
$q^m (1-\sqrt{(q^m -\delta(\mathrm{RM}_q(s,m))) /q^m})$ errors.

\begin{example}
We have $\mathcal C_1=\mathrm{Hyp}_9(9,2) \subseteq C_2 = \mathrm{RM}_9(s,2),$ where $s\geq 12$. Note that $\delta(C_2) = 5,$ while $\delta(\mathcal C_1) = 9$. Using the algorithm explained in this section and Pellikkan-Wu's decoder for $\mathcal C_2$, we can correct up to $2$ errors in $\mathcal C_1$ (which coincides with the error correcting capability of $\mathcal C_2$), while the error-correcting capability of the code $\mathcal C_1$ is $t_{\mathcal C_1} 
=4.$
\end{example}

\subsection{Decoding algorithm based on the largest Reed-Muller code}
\label{sub:dec2}
 Take $A\subseteq B\subseteq \mathbb N^m.$ We consider the set 
$\mathcal Q = \left\{ f \in \mathbb F_q[\mathbf X] \mid  \mathrm{supp}(f) \subseteq B \setminus A\right\},$
where the support of a polynomial 
$f\in \mathbb F_q[\mathbf X]$ is defined as $$\mathrm{supp}(f) = \left\{ \left(i_1, \ldots, i_m\right) \mid f=\sum_{j=1}^m \alpha_{i_j}\mathbf X^{i_j}, \alpha_{ i_j}\in \mathbb F_q \setminus \{0\}\right\}.$$ Observe that $|\mathcal Q| = q^{|B\setminus A|}$.

Let $\mathbf y\in \mathbb F_q^n$ be a received word. The following proposition tells us that, if the number of errors with respect to $\mathcal C_B$ is at most its error-correcting capability, i.e. $d_H(\mathbf y, \mathcal C_B) \leq t_{\mathcal C_{B}}$, then, there exists a unique polynomial $f\in \mathcal Q$ such that the nearest codeword to $\mathbf y - \mathrm{ev}_{\mathcal P}(f)$ is unique in $\mathcal C_A$.

\begin{proposition}
\label{Prop:Q}
Let $\mathbf y\in \mathbb F_q^n$ be a received word. Then there exists a polynomial $f\in \mathcal Q$ such that $d_H(\mathbf y-\mathrm{ev}_{\mathcal P}(f), \mathcal C_A) = d_H(\mathbf y, \mathcal C_B)$. Moreover, if $d_H(\mathbf y, \mathcal C_B)\leq t_{\mathcal C_B},$ the polynomial $f$ is unique.
\end{proposition}

\begin{proof}
We first prove the existence of such polynomial. Set $t := d_H(\mathbf y, \mathcal C_B)$, then there exists a codeword $\mathbf z=\mathrm{ev}_{\mathcal P}(g) \in \mathcal C_B$, with  $\mathrm{supp}(g) \subseteq B$, such that $d_H(\mathbf y, \mathbf z) = t$. Writing $g = f + \tilde{f},$ with $\mathrm{supp}(f) \subseteq B\setminus A$ and $\mathrm{supp}(\tilde{f}) \subseteq A$, we have $t=d_H(\mathbf y, \mathbf z) = d_H(\mathbf y - \mathrm{ev}_{\mathcal P}(f), \mathrm{ev}_{\mathcal P}(\tilde{f})),$
 with $f \in \mathcal Q$  and $\mathrm{ev}_{\mathcal P}(\tilde{f}) \in \mathcal C_A$.

Now we will prove the uniqueness when $t \leq t_{\mathcal C_B}$. Let $h\in \mathcal Q$ be another possible option. There exist two  error-vectors $\mathbf e_1, \mathbf e_2 \in \mathbb F_q^n$ with weight smaller or equal to $t_{\mathcal C_B},$ such that 
$\mathrm{ev}_{\mathcal P}(f) + \mathbf e_1 = \mathrm{ev}_{\mathcal P}(h) + \mathbf e_2 = \mathbf y.$
Therefore $\mathrm{ev}_{\mathcal P}(f) - \mathrm{ev}_{\mathcal P}(h) = \mathbf e_2 - \mathbf e_1$ is an element of $\mathcal C_B,$ with weight at most $2t_{\mathcal C_B} < \delta (\mathcal C_B)-1.$ Thus, the difference $\mathrm{ev}_{\mathcal P}(f) - \mathrm{ev}_{\mathcal P}(h)$ must be zero and $f = h$.
\qed
\end{proof}

Let $\mathcal C_A\subseteq \mathcal C_B$ monomial codes such that $A\subseteq B \subseteq \mathbb N^m.$ Let $\mathrm{Dec}_{\mathcal C_A}$ be a decoding algorithm for $\mathcal C_A,$ which corrects up to $E$ errors. By  Proposition \ref{Prop:Q}, we can define the following decoding algorithm for $\mathcal C_B$ that corrects also up to $E$ errors.


\noindent{\bf Initialization:} Let $\mathbf y\in \mathbb F_q^n$ be the received word. For each $f \in \mathcal Q$ we follow these steps \newline
\noindent {\bf Step 1.} Compute $\mathbf y - \mathrm{ev}_{\mathcal P}(f)$.\newline
\noindent {\bf Step 2.} Decode using $\mathrm{Dec}_{\mathcal C_A}$ the word $\mathbf y - \mathrm{ev}_{\mathcal P}(f)$.\newline
\noindent {\bf Step 3.} Denote by $L$ the output list of {\bf Step 2}, that is $L = \{\mathbf c \in \mathcal C_A \, \vert \, d_H(\mathbf c, \mathbf y - \mathrm{ev}_{\mathcal P}(f))\leq E\}.$ If $L$ is not empty, then for each $\mathbf c_A \in L$, we add to the output list $\mathbf c_A + \mathrm{ev}_{\mathcal P}(f)$.

The previous list-decoding algorithm for $\mathcal C_B$ corrects up to $E$ errors and requires $q^{|B\setminus A|}$ calls to $\mathrm{Dec}_{\mathcal C_A}$. Moreover, if $E \leq t_{\mathcal C_B}$, one can easily transform the previous list-decoding algorithm into the following unique decoding one:

\noindent{\bf Initialization:} Let $\mathbf y\in \mathbb F_q^n$ be the received word. Assume that the number of errors is at most $E \leq t_{\mathcal C_B}$.\newline
\noindent {\bf Step 1} Compute $\mathbf y - \mathrm{ev}_{\mathcal P}(f)$ for some $f\in \mathcal Q$.\newline
\noindent {\bf Step 2} Decode using the decoder $\mathrm{Dec}_{\mathcal C_A}$ the word $\mathbf y - \mathrm{ev}_{\mathcal P}(f)$.\newline
\noindent {\bf Step 3} If the result of {\bf Step 2} is codeword $\mathbf c_A$ such that 
$\mathrm w_H(y-ev_P(f), \mathbf c_A)\le t_{\mathcal C_B}$, then return $\mathbf c_A + \mathrm{ev}_{\mathcal P}(f)$.\newline
\noindent {\bf Step 4} Otherwise, go back to {\bf Step 1}. 

The correctness of this algorithm is justified by Proposition \ref{Prop:Q} and the fact that $E \leq t_{\mathcal C_B} \leq t_{\mathcal C_A}$ and, hence, the output in {\bf Step 2} has at most one element.

The algorithms proposed involve at most $|\mathcal Q| = q^{|B \setminus A|}$ calls to $\mathrm{Dec}_{\mathcal C_A}$, which could be interpreted as an inefficient algorithm from a theoretical point of view. Nevertheless, for practical purposes, if the difference between the sets $B$ and $A$ is small, this algorithm defines an efficient algorithm for $\mathcal C_B$ that corrects up to $E$ errors, as long as an efficient algorithm for $\mathcal C_A$ exists and corrects the same number of errors.

\subsection{Intermediate case.}
\label{sub:dec3}
In this section we are going to study an intermediate proposal between the two previous options. We will do our study for the case of two variables.

\begin{proposition}
\label{prop:intermediate}
Let $s$ be the smallest integer such that $\mathrm{Hyp}_q(d,2) \subseteq \mathrm{RM}_q(s,2)$. Let $H,R \subseteq \mathbb N^2$ such that $\mathcal C_H = \mathrm{Hyp}_q(d,2)$ and $\mathcal C_R = \mathrm{RM}_q(s-1,2)$. Then
\[|H-R|\leq 2 ( \sqrt[4]{d} + 1).\]
\end{proposition}

\begin{proof}
We have that $\mathcal C_H = \mathrm{Hyp}_q(d,2) \subseteq \mathrm{RM}_q(s,2)$. An easy observation is that whenever $(i_1,i_2) \in H - R$, then  $i_1 + i_2 = s$. Thus, 
\begin{equation} H - R = \{(t, s-t) \, \vert \, t \in \lii 0, s \rii {\rm\ and\ } (q- t) (q - s + t) \geq d\}. \end{equation}
Hence we are looking for those $t \in \lii 0, s \rii$ such that $(q- t) (q - s + t) \geq d$, or equivalently, 
\[\left( q- \frac{s}{2} \right)^2 - \left(\frac{s}{2} - t \right)^2 \geq d \Rightarrow \left|\frac{s}{2}-t\right|\leq \underbrace{\sqrt{\left( q-\frac{s}{2}\right)^2 -d}}_{\Delta}.\]
If we compute $\Delta$, then: $|H- R|\leq 2\Delta + 1$.
We separate the proof in two cases depending on the parity of $s$. Recall that, by \cite[Proposition 3.2]{our}, we know that $s=\left \lfloor 2(q-\sqrt{d})\right \rfloor$.

\begin{enumerate}
    \item If $s+1=2r$, then $(r,r) \notin H$. As a consequence, $(q-\frac{s+1}{2})^2 < d$. Or equivalently, $\left(q-\frac{s}{2} \right)^2 - q + \frac{2s+1}{4}< d$.  Thus,
    $$\Delta \leq \sqrt{q-\frac{2q+1}{4}}< \sqrt{q-\frac{4(q-\sqrt{d})-1}{4}} = \sqrt{\sqrt{d}+\frac{1}{4}} < \sqrt[4]{d}+\frac{1}{2}, \hbox{ for } d\geq 1.$$
    
    \item If $s=2r$, then $(r,r+1)\notin H$. As a consequence,
    $\left(q-\frac{s}{2}\right)\left(q-\frac{s}{2}-1\right)< d$. Or equivalently, $\left(q-\frac{s}{2} \right)^2 - \left( q-\frac{s}{2}\right)<d$. Thus,
    $$\Delta \leq \sqrt{q-\frac{s}{2}}< \sqrt{q-\frac{2(q-\sqrt{d})-q}{2}} = \sqrt{\sqrt{d}+\frac{1}{2}}< \sqrt[4]{d}+\frac{1}{2}\hbox{ for } d\geq 1.$$
\end{enumerate}
\qed
\end{proof}

Now we have the ingredients to define a new decoding algorithm for
the code $\mathcal C_H=\mathrm{Hyp}_q(d,2)$. Let $s$ be the smallest integer such that $\mathrm{Hyp}_q(d,2)\subseteq \mathrm{RM}_q(s,2)$. See \cite[Proposition 3.2]{our} for a precise description of such parameter $s$. We define $R\subseteq \mathbb N^2$ such that $\mathcal C_R=\mathrm{RM}_q(s-1,2)$. By Proposition \ref{prop:intermediate} we know that $|H- R|\leq c \sqrt[4]{d}$. 
Let $\mathrm{Dec}_{R}$ be a decoding algorithm for $\mathcal C_R$ that corrects up to $t_R=\left \lfloor \frac{d(\mathcal C_R)-1}{2}\right \rfloor$ errors. Unifying the ideas of the above decoding algorithms (see sections \ref{sub:dec1} and \ref{sub:dec2}), we have a decoding algorithm for $\mathcal C_H$ that corrects up to $E$ errors and requires $q^{c\sqrt[4]{d}}$ calls to $\mathrm{Dec}_{R}$. If we compare it with the algorithm proposed in Section \ref{sub:dec1}, we have a poorer complexity but we can correct more errors. 
This approach is particulary well suited when  $d'=\delta(\mathrm{RM}_q(s,2)) \geq q$ and $|H-R|$ is small. In this case, $\delta(\mathcal C_R) = d' - q + 1$ and, as a consequence, the correction capability of the auxiliary Reed-Muller code we are using to decode is increased by around $q/2$. 

\begin{example}
Consider $\mathcal C= \mathrm{Hyp}_q(d,2)$ with $q=11,$ and $d=32.$ See Figure \ref{example-section3.3} for a representation of this example. Throughout this example, we use Pellikaan-Wu's list-decoder (PW) for Reed-Muller codes.
By \cite[Proposition 3.2, Proposition 4.2]{our}   we have that 
  $\mathrm{RM}_q(s',2) \subseteq \mathcal C \subseteq \mathrm{RM}_q(s,2)$   for $s'\leq 8$ and  $s\geq 10.$  
\begin{enumerate}

\item Thus, $\mathcal C \subseteq \mathcal C_{R_1} = \mathrm{RM}_q(10,2)$. Using Section \ref{sub:dec1} and applying PW algorithm on $\mathcal C_{R_1}$ we have an efficient decoding algorithm for $\mathcal C$ that \textbf{corrects up to $5$ errors} and requires just one call to $\mathrm{Dec}_{\mathcal C_{R_1}}$.

\item Thus, $\mathcal C_{R_2} = \mathrm{RM}_q(8,2) \subseteq \mathcal C$. 
Using Section \ref{sub:dec2} to decode $\mathcal C$ and applying PW algorithm on $\mathcal  C_{R_2}$ we have a decoding algorithm for $\mathcal C$ that {\bf corrects up to $16$ errors}. This algorithm uses the decoder $\mathrm{Dec}_{\mathcal  C_{R_2}}$ plus some brute force. In particular, it requires $q^{11}$ calls to $\mathrm{Dec}_{\mathcal  C_{R_2}},$ where $11 = |H-R_2|$, and $H$ and $R_2$ are the sets in $\mathbb N^2$ that define $\mathcal C$ and $\mathcal  C_{R_2}$, respectively.

\item The intermediate proposal between the two previous options is described in Section \ref{sub:dec3}. More precisely, we consider $\mathcal  C_{R_3} =  \mathrm{RM}_q(9,2)$ and we perform $q^{|H-R_3|} = q^5$ calls to the PW decoder for $\mathcal C_{R_3}$ to {\bf correct up to $10$ errors} in $\mathcal C$, where $R_3$ is the set in $\mathbb N^2$ that defines $\mathcal  C_{R_3}$.

\end{enumerate}

\begin{figure}[ht]
 \centering{ \includegraphics[scale=0.35]{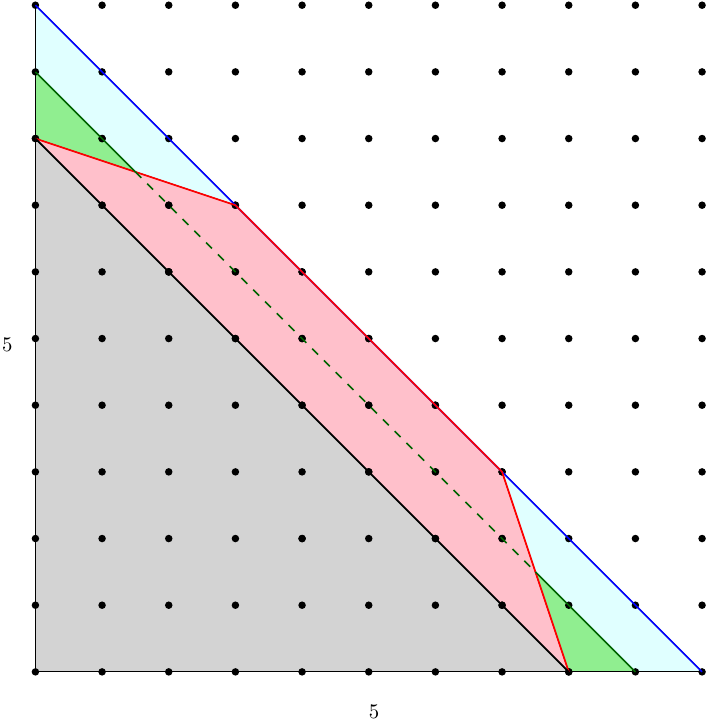}}
  \caption{Let $q=11$ and $m=2$. In this Figure the code $\mathrm{Hyp}_{q}(32,2)$ is equal to $\mathcal C_H$, $\mathrm{RM}_q(10,2)$ is equal to $\mathcal C_{R_1}$, $\mathrm{RM}_q(8,2)$ is equal to $\mathcal C_{R_2}$ and the code $\mathrm{RM}_q(9,2)$ equals to $\mathcal C_R,$ where $H$, $R_1$, $R_2$ and $R$ are the sets of lattice points below the red, the blue, the black and the green curve, respectively.}
  \label{example-section3.3}
  \end{figure}
\end{example}


\section{Decoding  based on tensor products of RS codes}\label{21.04.05}

Let $f\in \mathbb F_q[X_1, \ldots, X_m]$ be a polynomial. The maximum degree of $f$ with respect to $X_i$ is denoted by $\deg_{X_i}(f).$ Now we define a family of monomial codes that we call cube codes of order $s$, which consist of the evaluation of the polynomials $f\in \mathbb F_q[X_1,\ldots,X_m]$ that satisfy that $\mathrm{deg}_{X_i}(f)\leq s$ for each $1\le i\le m$, on the $q^m$ points of $\mathbb F_q^m.$ The formal definition is as follows.

\begin{definition}[Cube codes]
Take $s\in \mathbb N$ and define $A= \lii 0, s\rii^m.$ The monomial code $\mathcal C_A,$ denoted by $\mathrm{Cube}_q(s,m),$ is called the \emph{cube code} over $\mathbb F_q$ of order $s\geq 0$ with $m\geq 1$ variables.
\end{definition}
A Reed-Solomon code can be seen as a Reed-Muller code $\mathrm{RM}_q(d,m)$ of order $d$ and $m=1$ variables. That is, the Reed-Solomon code of order $s$ is defined as
$$\mathrm{RS}_q(s)=\left\{ \mathrm{ev}_{\mathcal P}(f) \mid f\in \mathbb F_q[X] \hbox{ and } \deg(f)\leq s\right\} = \mathrm{RM}_q(s,1).$$

Reed-Solomon codes are one of the most popular and important families of codes. They are maximum distance separable (MDS) codes, thus a $\mathrm{RS}_q(s)$ is a code with parameters $[q,s+1,q-s]_q.$ Reed-Solomon codes have efficient decoding algorithms. In the literature, the two primary decoding algorithms for Reed-Solomon codes are the Berlekamp-Massey algorithm \cite{Be}, and the Sugiyama et al. adaptation of the Euclidean algorithm \cite{Su}, both designed to solve a key equation.

\begin{remark}
Asume that the polynomials that define the Reed-Solomon code $\mathrm{RS}_q^i(s)$ belong to $\mathbb{F}_q[X_i].$ It is easy to see that the cube code $\mathrm{Cube}_q(s,m)$ is the tensor product of the $m$ Reed-Solomon codes $\mathrm{RS}_q^1(s), \ldots, \mathrm{RS}_q^m(s).$ In other words, we have that
\[\mathrm{Cube}_q(s,m) = \mathrm{RS}_q^1(s) \otimes \cdots \otimes \mathrm{RS}_q^m(s).\]
\end{remark}
\begin{proposition}
The minimum distance of the cube code $\mathrm{Cube}_q(s,m)$ coincides with its footprint bound for the deglex monomial ordering. Therefore, the code  $\mathcal C=\mathrm{Cube}_q(s,m)$ has length $n(\mathcal C)=q^m$, dimension $k(\mathcal C) = (s+1)^m,$ and minimum distance $\delta(\mathcal C) = (q-s)^m$.
\end{proposition}
\begin{proof}
This is a consequence of the fact that the cube code is the tensor product of Reed-Solomon codes.
\qed
\end{proof}

\begin{proposition}
\label{prop:highest-cube}
Take $d\in \mathbb N.$ Then,
\begin{itemize} \item[{\rm (a)}] $\mathrm{Cube}_q(s,m) \subseteq \mathrm{Hyp}_q(d,m)$ if and only if $s \leq q - \sqrt[m]{d}$.
\item[{\rm (b)}] 
$\mathrm{Hyp}_q(d,m) \subseteq \mathrm{Cube}_q(s',m)$ if and only if $s' \geq  q - \left\lceil \frac{d}{q^{m-1}} \right\rceil.$
\end{itemize}

\end{proposition}
\begin{proof}
Let $H \subset \lii 0,q-1 \rii^m$ such that $\mathcal C_H =  \mathrm{Hyp}_q(d,m)$. If $\mathbf i = (i_1,\ldots,i_m)$ satisfies that $0 \leq i_j \leq q - \sqrt[m]{d}$ for all $1\le j\le m$, then $\prod_{j = 1}^m (q - i_j) \geq d.$ Hence, $\mathrm{Cube}_q(s,m) \subseteq \mathrm{Hyp}_q(d,m)$ for all $s \leq q - \sqrt[m]{d}$. If $s > q - \sqrt[m]{d}$, then $(q-s)^m < d.$ Therefore $(s,\ldots,s) \notin H$ and $\mathrm{Cube}_q(s,m)  \not\subseteq \mathrm{Hyp}_q(d,m)$.

 We now check that $\mathrm{Hyp}_q(d,m) \subseteq \mathrm{Cube}_q(s',m)$ if and only if    $s' \geq r :=  q - \left\lceil \frac{d}{q^{m-1}} \right\rceil$. For $s' \geq r$, it suffices to observe that for all $\mathbf i = (i_1,\ldots,i_m) \in H$, we have that $\max\{i_j\} \leq r$. Otherwise $\prod_{j =1}^m (q-i_j) \leq q^{m-1} (q-\max\{i_j\}) < q^{m-1}(q-r) \leq d$, a contradiction. Moreover, we have that $(r,0,\ldots,0) \in H$, so if $s' < r$ then $\mathrm{Hyp}_q(d,m) \not\subseteq \mathrm{Cube}_q(s',m)$.
 \qed
\end{proof}

\begin{theorem}
\label{Prop:dec_cube}
Let $\mathrm{Dec}_{\mathrm{RS}}$ be a decoding algorithm for $\mathrm{RS}_q(s)$ that corrects up to $t_{\mathrm{RS}}$ errors. Then, there exists a decoding algorithm $\mathrm{Dec}_{\mathrm{Cube}_q(s,m)}$ for $\mathrm{Cube}_q(s,m)$ that corrects up to $(t_{\mathrm{RS}}+1)^m-1$ errors and requires calling $f(m)$ times the decoding algorithm $\mathrm{Dec}_{\mathrm{RS}},$ where 
$$f(m)  = \sum_{i=0}^{m-1} (s+1)^{m-1-i} q^{i} \leq m q^{m-1} = \frac{mn}{q} \hbox{, for all } m \geq 1.$$
\end{theorem}

\begin{proof}

We denote by $\alpha_1,\ldots,\alpha_q$ all the elements of $\mathbb F_q$.
We proceed by induction on $m\in \mathbb N$. 
For $m=1$, since $\mathrm{Cube}_q(s,1) = \mathrm{RS}_q(s)$ then, there exists $\mathrm{Dec}_{\mathrm{RS}}$ that corrects up to $t_{\mathrm{RS}}$ errors.

Now assume that there exists a decoding algorithm for $\mathrm{Cube}_q(s,m-1)$ that corrects up to $(t_{\mathrm{RS}}+1)^{m-1}-1$ errors. Without loss of generality we reorder the points $\mathcal P=\{ P_1, \ldots, P_n\} = \mathbb F_q^m$ with $n=q^m$ in such a way that the first $q^{m-1}$ points of $\mathcal P$ are those that have $\alpha_1$ in their first coordinate, then those that have $\alpha_2$, and so on. Let $\mathbf v = (\mathbf v_1, \ldots , \mathbf v_q)\in \mathbb F_q^n$ where $\mathbf v_i \in \mathbb F_q^{q^{m-1}}$ be such that there exists $\mathbf u = (\mathbf u_1, \ldots, \mathbf u_q)\in \mathrm{Cube}_q(s,m)$ where $\mathbf u_i \in \mathbb F_q^{q^{m-1}}$ with $d(\mathbf v, \mathbf u)< (t_{RS}+1)^m$. As $\mathbf u \in \mathrm{Cube}_q(s,m)$, there exists
$$f(X_1, \ldots, X_m) = \sum_{i_1, \ldots, i_m = 0}^s \beta_{i_1, \ldots, i_m} X_1^{i_1} \cdots X_m^{i_m}\in \mathbb F_q[X_1, \ldots, X_m],$$ such that $\mathbf u = \mathrm{ev}_{\mathcal P}(f)$ and $\mathbf u_i = \mathrm{ev}_{\mathcal P'}(f(\alpha_i, X_2, \ldots, X_m))$ with $\mathcal P' = \mathbb F_q^{m-1}$. We are going to show how to recover $f(X_1, \ldots, X_m)$ from $\mathbf v$ by calling $q$ times the decoder $\mathrm{Dec}_{\mathrm{Cube}_q(s,m-1)}$ and $(s+1)^{m-1}$ times the decoder $\mathrm{Dec}_{\mathrm{RS}}$.

For all $i \in \{1, \ldots, q\}$, we define $d_i = d(\mathbf u_i, \mathbf v_i)$. We say that $i \in \{1, \ldots, q\}$ is GOOD if $d_i < (t_{RS}+1)^{m-1}$; otherwise we say that $i$ is BAD. Let $d_{bad}$ be the number of BAD values $i \in \{ 1, \ldots, q\}$. Since
 \[ (t_{RS}+1)^m > \sum_{i = 1}^q d_i \geq \sum_{i {\rm \ is \ BAD}} d_i \geq (t_{RS}+1)^{m-1} d_{bad},\] we have that $d_{bad} < (t_{RS}+1)$.
 
Write $f(X_1, \ldots, X_m) = \sum_{j_2, \ldots, j_m=0}^s h_{j_2, \ldots, j_m}(X_1) X_2^{j_2} \cdots X_m^{j_m}$ where the univariate polynomial $h_{j_2, \ldots, j_m}(X_1)\in \mathbb F_q[X_1]$ has degree at most $s$, for all $j_2\ldots, j_m \in \{0, \ldots, s\}$. For all $\ell \in \{1,\ldots,q\}$, consider 
\begin{eqnarray*}
g_\ell(X_2, \ldots, X_m) & := & f(\alpha_\ell,X_2, \ldots, X_m) = \sum_{j_1, \ldots, j_m = 0}^s \beta_{j_1, \ldots, j_m}  {\alpha_\ell}^{j_1} X_2^{j_2} \cdots X_m^{j_m} =\\
& = &  \sum_{j_2, \ldots, j_m = 0}^s h_{j_2, \ldots, j_m}(\alpha_\ell) X_2^{j_2}\cdots X_m^{j_m} \in \mathbb F_q[X_2, \ldots, X_m],
\end{eqnarray*}
which   satisfyies that $\deg_{X_i}(g_{\ell}) \leq s$ for all $i \in \{2, \ldots, m\}$. Moreover, whenever $\ell$ is GOOD,  we can recover $g_\ell$ by means of $\mathrm{Dec}_{\mathrm{Cube}_q(s,m-1)}$ and the values   $\mathbf v_\ell=(v_{\ell_1},\ldots,v_{\ell_{q^{m-1}}})\in \mathbb{F}_q^{q^{m-1}}$.

Thus, if $\ell$ is GOOD, we have   recovered the value $h_{j_2, \ldots, j_m}(\alpha_\ell, X_2, \ldots, X_m)$ for all $j_2, \ldots, j_m \in \{0, \ldots, s\}$. As there are at least $(q-t_{RS})$ GOOD values, then for each $j_2, \ldots, j_m \in \{0, \ldots, s\}$ we have at least $(q-t_{RS})$ correct evaluations of $h_{j_2, \ldots, j_m}(X_1).$ Thus, by induction we can recover $\beta_{j_1, \ldots, j_m}$ using $(s+1)^{m-1}$ times the decoder $\mathrm{Dec}_{\mathrm{RS}}$.

Now, let $f(m)$ be the number of times that algorithm $\mathrm{Dec}_{\mathrm{Cube}_q(s,m)}$ calls algorithm $\mathrm{Dec}_{\mathrm{RS}}.$ We will deduce a formula for $f(m)$ by induction on $m\in \mathbb N$. First notice that for $m=1$, since $\mathrm{Cube}_q(s,1) = \mathrm{RS}_q(s),$ we have that $f(1) = 1$. Moreover, from the above paragraphs, we can deduce that 
  $ f(m) = q f(m-1) + (s+1)^{m-1}f(1). $ 
Now we assume that $f(r) = q f(r-1) + (s+1)^{r-1}f(1)$ for all $r\leq m$ and we try to show the result for $f(m)$. Indeed,
\begin{eqnarray*}
f(m) & = &  q f(m-1) + (s+1)^{m-1}f(1) \\
 & = & q ((s+1)^{m-2} + q f(m-2) ) + (s+1)^{m-1}\\
 & = & (s+1)^{m-1} + q(s+1)^{m-2} + q^2 f(m-2) \\
 & = & \cdots = (s+1)^{m-1} + q(s+1)^{m-2} + q^2(s+1)^{m-3} + \cdots + q^{m-1} f(1) \\
 & = & \sum_{i=0}^{m-1} (s+1)^{m-1-i} q^{i}.
\end{eqnarray*}
This completes the proof. \qed\end{proof}

\begin{example} 
Consider $\mathcal C= \mathrm{Hyp}_q(d,2)$ with $q=32,$ and $d=225.$ See Figure \ref{example-cube} for a representation of this example. We will give different decoding algorithms for $\mathcal C$ using all ideas proposed in this article so far.

\begin{enumerate} 
\item By \cite[Proposition 3.2]{our}, $\mathcal C \subseteq \mathrm{RM}_q(s,2),$ for $s\geq 34$. Therefore, using Section \ref{sub:dec1}, and applying PW algorithm on $\mathcal  C_{R_1} = \mathrm{RM}_q(34,2)$,  we have an efficient decoding algorithm for $\mathcal C$ that \textbf{corrects up to $11$ errors} and requires just one call to $\mathrm{Dec}_{\mathcal C_{R_1} }$.

\item By \cite[Theorem 4.3]{our}, $\mathrm{RM}_q(s,2)\subseteq \mathcal C,$ for $s\leq 24$. Therefore, using Section \ref{sub:dec2} and applying PW algorithm on $\mathcal C_{R_2} = \mathrm{RM}_q(24,2)$, we have a decoding algorithm for $\mathcal C$ that \textbf{corrects up to $127$ errors}   because the minimum distance of $RM_{32}(24,2)$ is 256 and so the error-correcting capability will be 127, which is beyond the error correction capability of $\mathcal C$. The algorithm uses the decoder $\mathrm{Dec}_{\mathcal C_{R_2}}$ plus brute force. In particular it requires $q^{156}$ calls to $\mathrm{Dec}_{\mathcal C_{R_2}}$ where $156 = |H-R_2|,$ and $H$ and $R_2$ are the sets in $\mathbb N^2$ defining $\mathcal C$ and $\mathcal C_{R_2}$, respectively.

\item The intermediate proposal between the two previous options described in Section \ref{sub:dec3} has a slight advantage in this example with respect to the first option. More precisely, one should consider $\mathcal C_{R'} = \mathrm{RM}_q(33,2)$ and perform $q^{|H-R'|} = q$ calls to the PW decoder for $\mathcal C_{R'}$ to {\bf correct up to $14$ errors} because the minimum distance of $RM_{32}(33,2)$ is 30 and so the error-correcting capability will be 14.

\item By Proposition \ref{prop:highest-cube}.(b), we have $\mathcal C\subseteq \mathrm{Cube}_q(s,2)$ for $s\geq 24$. We know an efficient decoding algorithm for $\mathcal C_3 = \mathrm{Cube}_q(24,2)$ that corrects up to 
$\displaystyle t_3 = (t_{\mathrm RS} + 1)^m-1 = 15,$  
where $t_{\mathrm RS} = \left\lfloor \frac{q-s-1}{2}\right\rfloor = 3$
denotes the error correction capability of $\mathrm{RS}_q(s)$ with $s=24$. Therefore, using Theorem \ref{Prop:dec_cube}, we have an efficient decoding algorithm for $\mathcal C$ that \textbf{corrects up to $15$ errors} and requires calling $q+s+1 = 57$ times the decoder $\mathrm{Dec}_{\mathrm{RS}}$.

\end{enumerate}

\begin{figure}[ht]
\centering{  \includegraphics[scale=0.40]{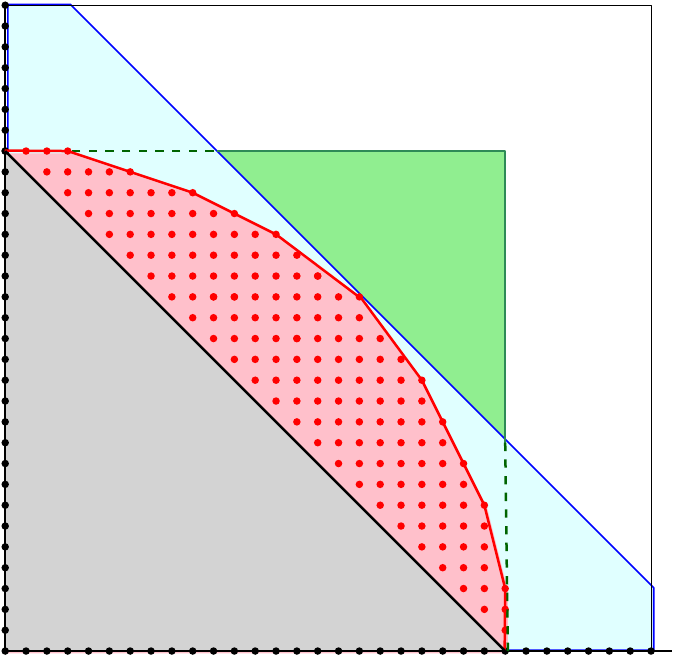}}
  \caption{Let $q=32$ and $m=2$. In this Figure the code $\mathrm{Hyp}_{q}(225,2)$ is equal to $\mathcal C_H$, $\mathrm{RM}_q(24,2)$ equals $\mathcal C_{R_1}$, the code $\mathrm{RM}_q(34,2)$ is equal to $\mathcal C_{R_2}$ and the code $\mathrm{Cube}_q(24,2)$ equals $\mathcal C_A,$ where $H$, $R_1$, $R_2$ and $A$ are the sets of lattice points below the red, the black, the blue and the green curve, respectively.}
  \label{example-cube}
  \end{figure}
\end{example}

\begin{remark}
All the ideas proposed in Section \ref{21.04.04} can be adapted to decode a hyperbolic code using a cube code. That is, by Proposition \ref{prop:highest-cube}, given a hyperbolic code $\mathrm{Hyp}$ we can find the largest (respectively smallest) cube code contained in (respectively that contains) $\mathrm{Hyp}$. And this result allow us to give decoding algorithms for hyperbolic codes in terms of the decoding algorithms of cube codes. Furthermore, an intermediate proposal between the two options can be used (as in Section \ref{sub:dec3}).
\end{remark}

\section{Decoding   based on a generalized Sudan's list decoding}\label{21.04.06}
 \begin{definition}\label{21.07.19}
    Let $\mathcal C_H=\mathrm{Hyp}_q(d,m)$ be a hyperbolic code with $H\subset\lii q-1\rii^m$. For $r\geq 1$, take $H_r\subset\lii q-1\rii^m$ such that $\mathcal C_{H_r}=\mathrm{Hyp}_q(r+1,m)$. For $i\geq 0,$ we define
$$L(d,r,i)=\{\mathbf{a}\in\lii q-1\rii^m\ |\ \mathbf{a}+iH\subseteq H_r\}.$$
\end{definition}

Observe that ${L(d,r,0)}=H_r$ for any $d$. We also have that $L(d,r,i+1)\subseteq L(d,r,i).$ Indeed, as
$L(d,r,i+1)+(i+1)H\subset H_r,$ then $L(d,r,i+1)\subset L(d,r,i+1)+H\subset L(d,r,i)$.

For the following algorithm, we assume that the numbers $r$ and $t$ satisfy the following conditions:
\begin{equation}
\sum_{i=0}^\infty \#L(d,r,i)>n \qquad \text{ and } \qquad t=\min\left\{i'\ |\ \sum_{i=0}^{i'} \#L(d,r,i)>n\right\}.
\label{eq:r,t}
\end{equation}

\noindent {\bf Notation:} Given $\mathbf{u},\mathbf{v}\in\mathbb{F}_q^n$, we define the Schur product as the component wise product on $\mathbb F_q^n$, i.e.
$(\mathbf{u}\ast\mathbf{v})_i=u_iv_i$ and $(\mathbf{u}^{\ast j})_i=u_i^j,$ for $j\geq 1$. We will write $\mathbf{u}^{\ast 0}$ for the word with all the components equal to $1$.

\noindent {\bf Initialization:} Let $\mathbf{y}\in\mathbb{F}_q^n$ be the received word, and define $r$ and $t$ according to the conditions \eqref{eq:r,t}.
\begin{enumerate}
\item[\bf Step 1] For $0\leq i\leq t,$ find $Q_i\in \mathbb F_q[L(d,r,i)]$, not all zero, such that $\sum_{i=0}^t\mathrm{ev}(Q_i)\ast\mathbf{y}^{\ast i}=0$.
\item[\bf Step 2] Factorize $Q(Y)=\sum_{i=0}^tQ_iY^i\in\mathbb{F}_q[\mathbf{X},Y]$ and detect all possible $f\in\mathbb{F}_q[\mathbf{X}]$ such that $(Y-f)|Q(Y).$  This can be done by the method of Wu \cite{Wu11}.
\item[\bf Step 3] Return $\{\mathrm{ev}(f)\in \mathcal C_H\ |\ f\text{ is a solution of Step 2}\}$.
\end{enumerate}

The list $\{\mathrm{ev}(f)\in\mathbb{F}_q[\mathbf{X}]\ |\ f\text{ is a solution of Step 2}\}$ is a list of at most $t$ elements that contains all the codewords $\mathbf{x}$ in $\mathcal C_H$ such that $d_H(\mathbf{x},\mathbf{y})\leq r$.

For completion, we write the proof of the last algorithm adapting it to the notation proposed in the previous lines.

\begin{theorem}{\bf (\cite[Theorem 4]{GM07})}
    The last algorithm gives the claimed output.
\end{theorem}

\begin{proof}
Since $\sum_{i=0}^t \#L(d,r,i)>n$, then the equation $\sum_{i=0}^t\mathrm{ev}(Q_i)\ast\mathbf{y}^{\ast i}=0$ has more indeterminates than equations and then a non-zero solution exists.
    Now, suppose that there exists $\mathbf{x}=\mathrm{ev}(f)\in \mathcal C_H$ such that $d_H(\mathbf{y},\mathbf{x})\leq r$. Since $f\in H$, then all the monomials in the support of $f^i$ are in $iH.$ Then, all the monomials $X^\mathbf{a}$ appearing in $Q_if^i$ satisfies $\mathbf{a}\in H_r$ by definition of $L(d,r,i)$. This implies that $\mathrm{ev}(Q(f))\in \mathcal C_{H_r}$ and 
    \begin{equation}
    \label{eq:circ}
    \mathrm{w}_H(\mathrm{ev}(Q(f)))\geq r+1\ \text{ or }\  Q(f)=0. 
    \end{equation}
    
    On the other hand, since $d_H(\mathbf{y},\mathbf{x})\leq r$, we have that $\sum_{i=0}^t\mathrm{ev}(Q_i)\ast\mathbf{y}^{\ast i}=0$ and $\mathrm{ev}(Q(f))$ can differ in at most $r$ distinct positions. Thus 
     $\mathrm{w}_H(\mathrm{ev}(Q(f)))\leq r.$ 
By Equation \eqref{eq:circ} we conclude that $Q(f)=0$, which means $(Y-f)|Q(Y)$. 
\qed
\end{proof}

To use the previous algorithm and to know the number of errors that we can correct, we just need to compute $L(d,r,i)$ along with their sizes. In general, it is not clear what is the form of $L(d,r,i),$ but with the following results we can estimate their sizes, which is one of the main contributions of this section.

\begin{proposition}
Let $\mathcal C_A=\mathrm{Hyp}_q(d_A,m)$ and $\mathcal C_B=\mathrm{Hyp}_q(d_B,m)$. Then
  $\mathcal C_{A+B}\subset\mathrm{Hyp}_q(d_A+d_B-q^m,m).$  
\end{proposition}

\begin{proof}
Take $f \in \mathbb F_q[A]$ and $g \in \mathbb F_q[B]$. Denote the set of zeros of $f$ (resp. g) in $\mathbb{F}_q^m$ by $Z(f)$ (resp. $Z(g)$). We know that $|Z(f)|\leq q^m-d_A$ and $|Z(g)|\leq q^m-d_B.$ This implies that
$|Z(fg)|\leq |Z(f)|+|Z(g)|\leq 2q^m-d_A-d_B.$ In other words, for any $fg\in \mathcal C_{A+B},$ we have that
$\mathrm{w}_H(\mathrm{ev}(f)\ast\mathrm{ev}(g))\geq d_A+d_B-q^m.$ Then
 $\delta(\mathcal C_{A+B})\geq d_A+d_B-q^m.$
As $\mathcal C_{A+B}$ is a monomial code, its minimum distance is the minimum of the footprints of its defining monomials. Thus we obtain the conclusion.
\qed
\end{proof}

With the above result we can bound the size of the set $L(d,r,1)$ and so, we can bound the number of errors that we can correct with the proposed algorithm when we adapt it to unique decoding.

\begin{corollary}
    Let $\mathcal C=\mathrm{Hyp}_q(d,m)$ and $d\geq r\geq 1$. If $\mathcal C_{H_1}=\mathrm{Hyp}_q(q^m+r-d+1,m)$, then
     $H_1\subseteq L(d,r,1).$ 
\end{corollary}

The number of errors that we can uniquely decode with the proposed algorithm is given by an easy-to-check formula.

\begin{corollary}
    If $\# H_r+\#H_1> n$, then the algorithm can correct up to $r$ errors solving a linear equation in $(\mathbb{F}_q[\mathbf{X}])[Y].$
\end{corollary}

If $m=2$ we can do even better. In the case of two variables, we can not only bound the size of the set $L(d,r,1)$ but we can also know exactly what monomial code is associated with such subset.
\begin{proposition}
\label{sec5:L1}
    Let $\mathcal C_H=\mathrm{Hyp}_q(d,2)$, with $d>q$. Take $a=\left \lfloor q-\frac{d}{q}\right \rfloor$ and $b=q-\frac{d}{q-a}$. For $r<a-b+1$, we have
    $\mathcal C_{L(d,r,1)}=\mathrm{Cube}_q(q-1-a,2).$ 
\end{proposition}

\begin{proof}
    Take $c=q-1-a.$ Observe that $(0,a)\in H$ but for any $a< i_2\in\mathbb{Z}$, $(0,i_2)\notin H$. Similarly, $(q-a)(q-b)\geq d$ but for any $i_1 > b$, $(q-a)(q-i_1)<d$. As $$(q-c-a)(q-c-b)=a-b+1>r,$$ and since $\{(i_1,i_2)\ |\ i_1+i_2\leq a+b\}$ and $H_r$ are both convex sets, then for any $(i_1,i_2)\in H$ such that $i_1+i_2\leq a+b$, we have $(i_1+c,i_2+c)\in H_r$. 
    
Now, suppose that $(i_1,i_2)\in H$ but $i_1+i_2>a+b$. Then we have
    \begin{align*}
        (q-i_1-c)(q-i_2-c)&=(q-i_1)(q-i_2)+c(i_1+i_2-2q)+c^2\\
        &\geq (q-a)(q-b)+c(i_1+i_2-2q)+c^2\\
        &> (q-a)(q-b)+c(a+b-2q)+c^2\\
        &=(q-a-c)(q-b-c)\\
        &=a-b+1\\
        &>r.
    \end{align*}
This means that $(i_1+c,i_2+c)\in H_r$. Then we have 
 $c+H\subset H_r.$
    By Definition~\ref{21.07.19}, we can easily see that if $(i_1,i_2)$ satisfies the property that $i_1,i_2\leq c,$ then $(i_1,i_2)\in L(d,r,1)$.
    
    Finally, we can assure $L(d,r,1)=\{(i_1,i_2)\ |\ i_1,i_2 \leq c\}$, otherwise it would exist some $(i_1,i_2)\in L(d,r,1)$ with $i_1>c$ such that $i_1+a>q-1$ and $(i_1+c,i_2+c)\in H_r,$ which is a contradiction.
    \qed
\end{proof}

The last result can be generalized for all the sets $L(d,r,i)$.

\begin{corollary}
    Let $\mathcal C_H=\mathrm{Hyp}_q(d,2)$, $d>q$, and $a,$ $b$ and $r$ as before. Then
    $\mathcal C_{L(d,r,i)}=\mathrm{Cube}_q(q-1-ia,2)$.
\end{corollary}

\begin{proof}
We know the case $i=1.$ Assume the result is true for $i\in\mathbb{N}$. As
$L(d,r,i+1)+H\subseteq L(d,r,i)=\mathrm{Cube}_q(q-1-ia,2)$ and $\mathrm{Hyp}_q(d,2)\subseteq\mathrm{Cube}_q(a,2)$ \cite[Theorem 4.3]{our}, then we have that for any $(i_1,i_2)\in H,$ $i_1,i_2 \leq a.$ This implies that
$$(q-1-(i+1)a+i_1,q-1-(i+1)a+i_2)\leq (q-1-ia,q-1-ia).$$
We conclude that $\{(i_1,i_2)\ |\ i_1,i_2\leq q-1-(i+1)a\}\subseteq L(d,r,i+1)$. The equality follows from the fact that $(a,0)$ is a point in $H$.
\qed
\end{proof}
\begin{remark}\rm
\label{sec5:rem}
Using the previous results, for $m = 2$ we have the following bound for the number of errors that our algorithm uniquely corrects. Take $t=\left \lfloor\frac{q-1}{a}\right \rfloor$ and $r<a-b+1,$ with $a$ and $b$ as before for $\mathcal C_H=\mathrm{Hyp}_q(d,2)$. If
$\#H_r+\sum_{i=1}^t(q-1-ia)^2>n,$ then the algorithm can correct up to $r$ errors. 
\end{remark}

\begin{example}
Consider $\mathcal C= \mathrm{Hyp}_q(d,2)$ with $q=16$ and $d=81$. See Figure \ref{example-section5} for a representation of this example. Take $r=8.$ Brute force computation on Geil and Matsumoto's algorithm gives that:
\begin{itemize}
\item $C_{L(d,r,0)} = \mathrm{Hyp}_q(r+1,2)$, which coincide with Definition \ref{21.07.19}.
\item Moreover, $C_{L(d,r,1)} = \mathrm{Cube}_q(5,2)$, which matches with Proposition \ref{sec5:L1} since
$$\mathrm{Cube}_q(5,2) = \mathrm{Cube}_q(q-1-a,2) \hbox{ with }a= \left\lfloor q-\frac{d}{q}\right\rfloor = 10.$$
\item Finally, $C_{L(d,r,2)} = \{0\}$.
\end{itemize}
Moreover, following Remark \ref{sec5:rem}, $r=8$ is the maximum number of errors that we can correct with this algorithm since 
$$r< a-b +1, \hbox{ where } a= \left\lfloor q-\frac{d}{q}\right\rfloor = 10 \quad \hbox{ and } \quad b=q-\frac{d}{q-a} = 2.5.$$
Observe that $\#H_r+\sum_{i=1}^t(q-1-ia)^2 = \#L(d,r,0) + \#L(d,r,1)>n = 256,$  where $t=\left \lfloor\frac{q-1}{a}\right \rfloor =1$ and $L(d,r,0)$ and $L(d,r,1)$ are the set of lattice points below the green and the black curve of Figure \ref{example-section5}, respectively. The sets $L(d,r,0)$ and $L(d,r,1)$ are, in general, difficult to describe. Using the results of this section we explicitly describe these sets when $m=2,$ and we provide  a lower bound on their sizes for the cases when $m>2$.
\begin{figure}[ht]
  \centering{\includegraphics[scale=0.25]{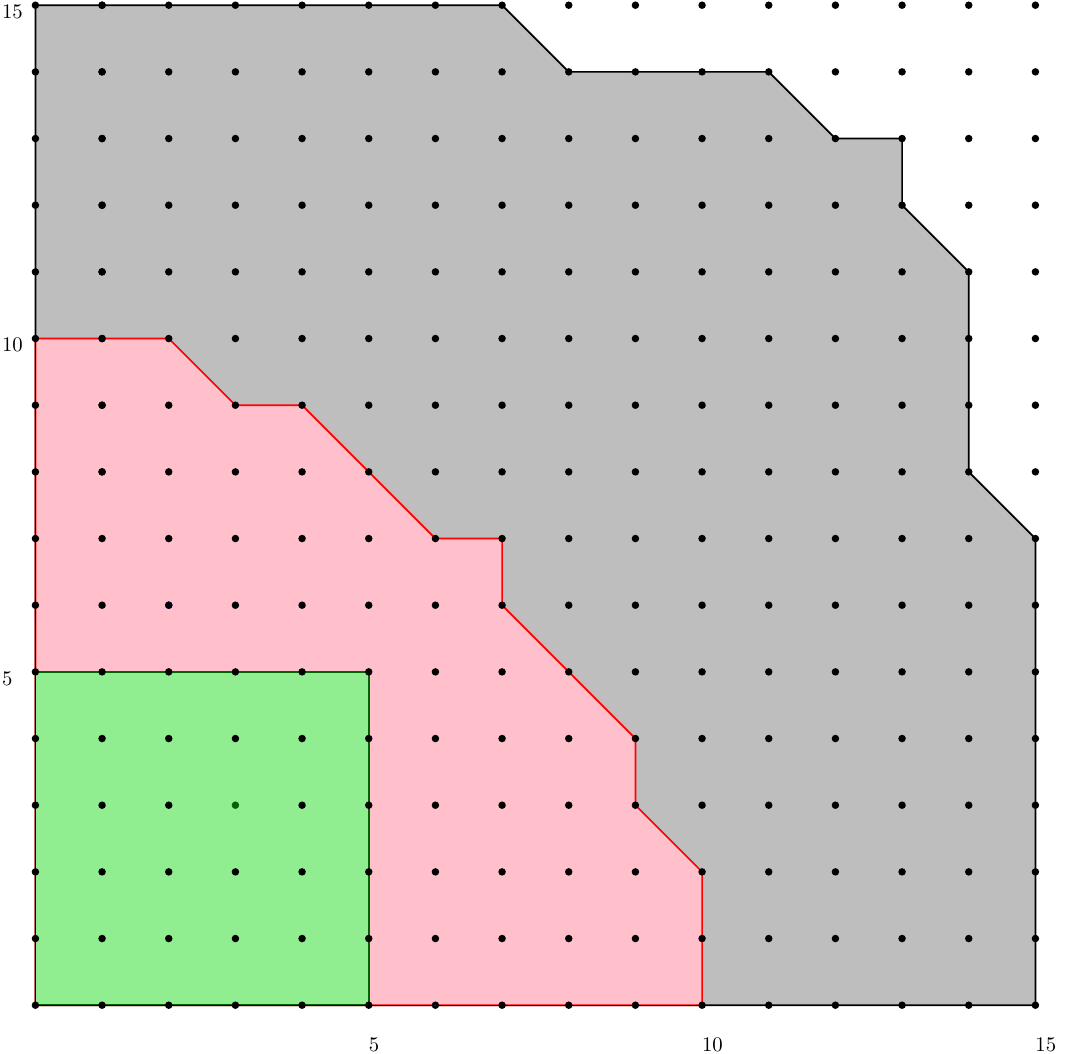}}
  \caption{Let $q=16$ and $m=2$. In this Figure the code $\mathcal C= \mathrm{Hyp}_{q}(81,2)$ is equal to $\mathcal C_H$, $\mathcal C_{L(81,8,0)} = \mathrm{Hyp}_q(9,2)$ is equal to $\mathcal C_{L_0}$ and $\mathcal C_{L(81,8,1)} = \mathrm{Cube}_q(5,2)$ is equal to $\mathcal C_{L_1}$ where $H$, $L_0$ and $L_1$ are the sets of lattice points below the red, the black and the green curve, respectively.}
  \label{example-section5}
  \end{figure}
\end{example}

\section{Comparisons and Conclusions}\label{sec:con}
 Table~\ref{tabla-final} compares the performance of the five decoding algorithms proposed in this paper for the hyperbolic code $\mathcal C = \mathrm{Hyp}_{q}(d,m),$ where $q=32, m=2$ and $d$ takes different values. Table~\ref{tabla-final} is composed by $6$ blocks, one for each value of $d.$ Each block contains 5 lines, which represent the following:
\begin{itemize}
\item First line refers to the algorithm of Section \ref{sub:dec1}. Here, we compute the smallest integer $s$ such that $\mathcal C  \subseteq  \mathrm{RM}_q(s,m)$. Then, we use   Pellikaan-Wu list-decoding algorithm for Reed-Muller codes to decode $\mathcal C$.

\item Second line refers to the algorithm of Section \ref{sub:dec2}. Here, we compute the largest integer $s$ such that $\mathrm{RM}_q(s,m) \subseteq \mathcal C$. Then, we use the Pellikaan-Wu list-decoding algorithm for Reed-Muller codes to decode $\mathcal C$ plus some brute force. 

\item Third line refers to the algorithm of Section \ref{sub:dec3}, an intermediate case between the above two options. In this case, we use again the Pellikaan-Wu list-decoding algorithm for Reed-Muller codes to decode $\mathcal C$.

\item Fourth line refers to the algorithm of Section \ref{21.04.05}. More precisely, we compute the smaller integer $s$ such that $\mathcal C  \subseteq  \mathrm{Cube}_q(s,m)$. Then, we use the algorithm described in Section \ref{21.04.05} for cube codes to decode $\mathcal C$. 

\item Fifth line refers to the specific algorithm known for $\mathcal C$ described in Section \ref{21.04.06}. 
\end{itemize}
The third column describes the number of calls to the corresponding decoder. The last column represents the minimum distance of the auxiliary code that we are using in each case.

\begin{table}[ht]
\begin{center}
\renewcommand{\arraystretch}{1.2}
\begin{tabular}{|c|c|c|c|c|}
\hline
\multirow{3}{*}{} & Error & \multirow{2}{*}{Number} & \multirow{2}{*}{Called} & \multirow{2}{*}{Minimum} \\
& correcting & \multirow{2}{*}{of calls} & \multirow{2}{*}{algorithm} & \multirow{2}{*}{distance} \\
& capability & & &\\ \hline \hline
\multirow{5}{*}{$d=257$} & 
$E_1 = 16$ & $1$ & $\mathrm{Dec}(\mathrm{RM}_{q}(31,2))$ & $\delta(\mathrm{RM}_{q}(31,2)) = 32$\\ \cline{2-5}
 & $E_2 = 155$ & $q^{134}$ & $\mathrm{Dec}(\mathrm{RM}_{q}(23,2))$ & $\delta(\mathrm{RM}_{q}(23,2)) = 288$\\ \cline{2-5}
& $E_3 = 32$ & $q^8$ & $\mathrm{Dec}(\mathrm{RM}_{q}(30,2))$ & $\delta(\mathrm{RM}_{q}(30,2)) = 64$\\\cline{2-5}
& $E_4 = 24$ & $56$ & $\mathrm{Dec}(\mathrm{RS}_{q}(23))$ & $\delta(\mathrm{Cube}_{q}(24,2)) = 81$\\ \cline{2-5}
& $E_5 = 23$ & $1$ & $\mathrm{Dec}(\mathrm{Hyp}_{q}(d,2))$ & $\delta(\mathrm{Hyp}_{q}(d,2)) = 257$\\ \hline \hline

\multirow{5}{*}{$d=225$} & 
$E_1 = 14$ & $1$ & $\mathrm{Dec}(\mathrm{RM}_{q}(34,2))$ & $\delta(\mathrm{RM}_{q}(34,2)) = 29$\\ \cline{2-5}
 & $E_2 = 137$ & $q^{156}$ & $\mathrm{Dec}(\mathrm{RM}_{q}(24,2))$ & $\delta(\mathrm{RM}_{q}(24,2)) = 256$\\ \cline{2-5}
& $E_3 = 15$ & $q$ & $\mathrm{Dec}(\mathrm{RM}_{q}(33,2))$ & $\delta(\mathrm{RM}_{q}(33,2)) = 30$\\\cline{2-5}
& $E_4 = 15$ & $57$ & $\mathrm{Dec}(\mathrm{RS}_{q}(24))$ & $\delta(\mathrm{Cube}_{q}(24,2)) = 64$\\ \cline{2-5}
& $E_5 = 19$ & $1$ & $\mathrm{Dec}(\mathrm{Hyp}_{q}(d,2))$ & $\delta(\mathrm{Hyp}_{q}(d,2)) = 225$\\ \hline \hline

\multirow{5}{*}{$d=193$} & 
$E_1 = 13$ & $1$ & $\mathrm{Dec}(\mathrm{RM}_{q}(36,2))$ & $\delta(\mathrm{RM}_{q}(36,2)) = 27$\\ \cline{2-5}
 & $E_2 = 118$ & $q^{182}$ & $\mathrm{Dec}(\mathrm{RM}_{q}(25,2))$ & $\delta(\mathrm{RM}_{q}(25,2)) = 224$\\ \cline{2-5}
& $E_3 = 14$ & $q^3$ & $\mathrm{Dec}(\mathrm{RM}_{q}(35,2))$ & $\delta(\mathrm{RM}_{q}(35,2)) = 28$\\\cline{2-5}
& $E_4 = 15$ & $58$ & $\mathrm{Dec}(\mathrm{Cube}_{q}(25,2))$ & $\delta(\mathrm{Cube}_{q}(25,2)) = 49$\\ \cline{2-5}
& $E_5 = 15$ & $1$ & $\mathrm{Dec}(\mathrm{Hyp}_{q}(d,2))$ & $\delta(\mathrm{Hyp}_{q}(d,2)) = 193$\\ \hline \hline

\multirow{5}{*}{$d=150$} & 
$E_1 = 12$ & $1$ & $\mathrm{Dec}(\mathrm{RM}_{q}(39,2))$ & $\delta(\mathrm{RM}_{q}(39,2)) = 24$\\ \cline{2-5}
 & $E_2 = 83$ & $q^{212}$ & $\mathrm{Dec}(\mathrm{RM}_{q}(27,2))$ & $\delta(\mathrm{RM}_{q}(27,2)) = 160$\\ \cline{2-5}
& $E_3 = 12$ & $q^{6}$ & $\mathrm{Dec}(\mathrm{RM}_{q}(38,2))$ & $\delta(\mathrm{RM}_{q}(38,2)) = 25$\\\cline{2-5}
& $E_4 = 8$ & $60$ & $\mathrm{Dec}(\mathrm{RS}_{q}(27))$ & $\delta(\mathrm{Cube}_{q}(27,2)) = 25$\\ \cline{2-5}
& $E_5 = 9$ & $1$ & $\mathrm{Dec}(\mathrm{Hyp}_{q}(d,2))$ & $\delta(\mathrm{Hyp}_{q}(d,2)) = 150$\\ \hline \hline 

\multirow{5}{*}{$d=65$} & 
$E_1 = 8$ & $1$ & $\mathrm{Dec}(\mathrm{RM}_{q}(47,2))$ & $\delta(\mathrm{RM}_{q}(47,2)) = 16$\\ \cline{2-5}
 & $E_2 = 49$ & $q^{343}$ & $\mathrm{Dec}(\mathrm{RM}_{q}(29,2))$ & $\delta(\mathrm{RM}_{q}(29,2)) = 96$\\ \cline{2-5}
& $E_3 = 8$ & $q^{6}$ & $\mathrm{Dec}(\mathrm{RM}_{q}(46,2))$ & $\delta(\mathrm{RM}_{q}(46,2)) = 17$\\\cline{2-5}
& $E_4 = 3$ & $62$ & $\mathrm{Dec}(\mathrm{RS}_{q}(29))$ & $\delta(\mathrm{Cube}_{q}(29,2)) = 9$\\ \cline{2-5}
& $E_5 = 5$ & $1$ & $\mathrm{Dec}(\mathrm{Hyp}_{q}(d,2))$ & $\delta(\mathrm{Hyp}_{q}(d,2)) = 65$\\ \hline \hline 

\multirow{5}{*}{$d=15$} & 
$E_1 = 3$ & $1$ & $\mathrm{Dec}(\mathrm{RM}_{q}(56,2))$ & $\delta(\mathrm{RM}_{q}(56,2)) = 7$\\ \cline{2-5}
 & $E_2 = 7$ & $q^{64}$ & $\mathrm{Dec}(\mathrm{RM}_{q}(48,2))$ & $\delta(\mathrm{RM}_{q}(48,2)) = 15$\\ \cline{2-5}
& $E_3 = 4$ & $q^{3}$ & $\mathrm{Dec}(\mathrm{RM}_{q}(55,2))$ & $\delta(\mathrm{RM}_{q}(55,2)) = 8$\\\cline{2-5}
& $E_4 = 0$ & $64$ & $\mathrm{Dec}(\mathrm{RS}_{q}(31))$ & $\delta(\mathrm{Cube}_{q}(31,2)) = 1$\\ \cline{2-5}
& $E_5 = 0$ & $1$ & $\mathrm{Dec}(\mathrm{Hyp}_{q}(d,2))$ & $\delta(\mathrm{Hyp}_{q}(d,2)) = 15$\\ \hline \hline

\end{tabular}
\end{center}
\caption{Comparison between the five different algorithms described above to decode $\mathrm{Hyp}_{32}(d,2),$ for different values of $d$.}
\label{tabla-final}
\end{table}

In Table~\ref{tabla-final} we observe that the algorithm with the greatest error correcting capability is always achieved by the method given in the second line. However, the huge amount of calls to the decoder makes it highly impractical. The third method always corrects more errors than the first one, but requires more calls to the decoder. Concerning the third, fourth and fifth method, we find instances where each of the methods outperforms the others. All the algorithms we propose except the fifth one, rely on a known decoder for either a Reed-Muller or a cube code. As a consequence, a better decoder for any of these codes would imply better error correction capability. Interestingly, when decoding in terms of cube codes, we reduce many times the problem to a code with a promisingly high minimum distance, but then we use a decoding algorithm with poor error correcting capability. We consider that it is an interesting problem to find better decoding algorithms for cube codes which, in particular, will lead to better decoding algorithms for hyperbolic codes.

\bibliographystyle{plain}

\end{document}